\documentclass{amsart}
\usepackage{amssymb}
\usepackage{amsfonts}
\usepackage{amsmath}
\usepackage{pstricks}
\usepackage{amscd}
\usepackage{graphicx}
\DeclareMathOperator{\sech}{{\rm sech}}
\DeclareMathOperator{\cosec}{{\rm cosec}}
\newtheorem{theorem}{Theorem}

\newtheorem{example}[theorem]{Example}
\newtheorem{lemma}[theorem]{Lemma}

\newtheorem{remark}[theorem]{Remark}
\begin{document}
\title{Evaluation of spherical GJMS determinants}
\author[T.~Mansour]{Toufik Mansour}
\address{T.~Mansour\\Department of Mathematics, University of Haifa, 3498838  Haifa, Israel}
\email{tmansour@univ.haifa.ac.il}
\author[J.S.~Dowker]{J.S. Dowker}
\address{J.S.~Dowker\\ Theory Group, School of Physics and Astronomy, The University of
Manchester, Manchester, England}
\email{dowker@man.ac.uk;\, dowkeruk@yahoo.co.uk}
\maketitle

\section*{Abstract}
An expression in the form of an easily computed integral is given for the determinant of the
scalar GJMS operator on an odd--dimensional sphere.  Manipulation yields a sum formula for
the logdet in terms of the logdets of the ordinary conformal Laplacian for other dimensions.
This is formalised and  expanded by an analytical treatment of the integral which produces
an explicit combinatorial expression directly in terms of the Riemann zeta function, and
$\log2$. An incidental byproduct is a (known) expression for the central factorial coefficients
in terms of higher Bernoulli numbers.
\smallskip

\noindent {\sc Keywords}: GJMS operator,  Riemann zeta function, higher Bernoulli numbers, determinants
\smallskip

\noindent {\sc 2010 Mathematics Subject Classification}: 58J52, 11M36

\section{Introduction}
This paper has the goal of presenting an easy--to--calculate formula for the determinant of
the GJMS conformally invariant operator, $P_{2k}$, on an odd--dimensional sphere.
Expressions for this are actually already known but it is always helpful to have alternative
representations if only for variety and checking. Furthermore, knowing the values on one
geometry allows one to find them on conformally related spaces, using Branson's $Q$
curvature if necessary.

Further motivation for computing these determinants follows from the AdS/CFT
correspondance in quantum field theory, one manifestation of which involves the partition
functions for conformally invariant higher spin fields on spheres, Tseytlin \cite{24}, whose
propagation operator takes a similar product form to the GJMS one (see the next section).

In \cite{1} a direct spectral evaluation of $\log\det P_{2k}$ was given  which yielded an
integral over a Plancherel measure and agreed with other derivations done using
dimensional regularisation in the context of the AdS/CFT correspondance by Diaz, \cite{2},
and Diaz and Dorn, \cite{3}, and via the Selberg $\zeta$--function by Guillarmou, \cite{4}.
The method here is an application of that in \cite{5} which involved the ordinary
conformally invariant Laplacian, $P_2$. A drawback is that k is restricted to be an integer
smaller than $d/2$, unlike some other formulae which allow a continuation in $k$.

\section{The general method}
It is necessary, briefly, to give some definitions and basic facts. Branson's construction of
$P_{2k}$, \cite{6}, in the special case of the (round) unit $d$--sphere is a simple product
(see also Graham \cite{23})
\begin{align}
P_{2k}&=\prod_{j=0}^{k-1}(B^2-\alpha_j^2),\qquad(\alpha_j=j+1/2)\notag\\
&=\frac{\Gamma(B+1/2+k)}{\Gamma(B+1/2-k)}=B^{[2k+1]-1},\label{eq1}
\end{align}
expressed as a central factorial.\footnote{ {\it e.g.} Steffensen, \cite{22}. It is amusing to
note that the product formula of Juhl, \cite[Lemma 6.1]{7}, is the central version of
Elphinstone's theorem of 1858, \cite{8}.} Here $B=\sqrt{P_2+1/4}$ with
$P_2=-\Delta_2+((d-1)^2-1)/4$ the Yamabe-Penrose operator, sometimes denoted by
$Y_d$, on the sphere.

In \cite{5}, for a different purpose, the determinant of the operator $B^2-\alpha^2$,
where $\alpha$ was a general parameter, was calculated. Because, as shown in \cite{1},
there is no multiplicative anomaly for odd dimensions\footnote{There is one for the
hemisphere}, the result can be used to find $\log\det P_{2k}$ easily as the sum,
\begin{align}
\log\det P_{2k}=\sum_{j=0}^{k-1}\log\det(B^2-\alpha_j^2).\label{eq2}
\end{align}
It will turn out very shortly that this sum can be done.

\section{The calculation}
The method involves a Bessel function representation for the $\zeta$--function of the
operator $B^2-\alpha^2$ on the odd $d$--sphere, for example see \cite{9,10}. A contour
technique then provides the expression, (see \cite[Equation (11)]{5}),
\begin{align}
\log\det (B^2-\alpha_j^2)&=-\frac{1}{2^{d-2}}\int_0^\infty dx\,{\rm Re}
\frac{\cosh(\tau/2)\cosh(\alpha_j\tau)}{\tau\sinh^d(\tau/2)}\notag\\
&=\frac{(-1)^{(d+1)/2}}{2^{d-2}}\int_0^\infty dx\,\frac{(-1)^j\pi}
{x^2+\pi^2}\frac{\sinh(x/2)\sinh(\alpha_jx)}{\cosh^d(x/2)}.\label{eq3}
\end{align}
On the first line $\tau= x + iy$, ($0\leq y\leq 2\pi$). The second line results from the choice
$y =\pi$ and the fact that $\alpha_j$ is a half–integer.

The geometric sum over $j$ in \eqref{eq2} gives,
\begin{align}
\log\det P_{2k}&=\frac{(-1)^{(d-1)/2+k}}{2^{d-1}}\int_0^\infty dx\frac{\pi}{x^2+\pi^2}\frac{\sinh(x/2)\sinh(kx)}{\cosh^{d+1}(x/2)}.\label{eq4}
\end{align}
which is the main computational result. Because of the factor $(-1)^k$ it works sensibly
only for $k$ an integer. Furthermore, the integral diverges if $2k >d$.\footnote{This is a
consequence of the appearance of negative eigenvalues of $B-\alpha_j$ but this will not be
remedied here.} So from now on we assume that $2k\leq d$.

\section{A Chebyshev rearrangement}
For $k$ an integer, there is the expression
$$\frac{\sinh(kx)}{\sinh(x/2)}=U_{2k-1}(\cosh(x/2)),$$
in terms of Chebyshev polynomials of the second kind and \eqref{eq4} reads,
\begin{align}
\log\det P_{2k}&=\frac{(-1)^{(d-1)/2+k}}{2^{d-1}}\int_0^\infty dx\frac{\pi}
{x^2+\pi^2}\frac{\sinh^2(x/2)U_{2k-1}(\cosh(x/2))}{\cosh^{d+1}(x/2)}.\label{eq5}
\end{align}
The dimension, $d$, has been indicated because expanding the Chebyshev polynomial will
produce a sum of ordinary $\log\det$s of $P_2(d')$ for dimension $d'$ in the range $d$ to
$d-2k + 2$.

Making this explicit, one has, first,
$$U_{2k-1}(x)=x(u_0+u_1x^2+\cdots+u_{k-1}x^{2k-2},$$
where the coefficients, $u_j(k)$, are known, both in general, for example see \cite{11,12},
and by recursion for a given $k$. (There are tables of them.) Then,
\begin{align}
\log\det P_{2k}&=\det P_2^{v_0}(d)\det P_2^{v_1}(d)\cdots P_2^{v_{k-1}}(d-2k+2),
\label{eq6}
\end{align}
where the powers $v_j(k)$ are simply related to the $u_j(k)$ by,
\footnote{They are the numbers in the square array of binomial coefficients read
 by anti–diagonals.}
$$v_j(k)=\frac{(-1)^{k-1+j}}{2^{2j+1}}u_j(k).$$
The coefficients $u_j$ alternate in sign so the $v_j$ all have the same sign, and,
moreover, are integers. It's best to give a few examples. Trivially, $v_0(1) = 1$ and,
non–trivially, one finds the determinant ''product rules",
\begin{align}\label{eq7}
P_4(d)&\sim P_2^2(d)P_2(d-2),\notag\\
P_6(d)&\sim P_2^3(d)P_2^4(d-2)P_2(d-4),\notag\\
P_8(d)&\sim P_2^4(d)P_2^{10}(d-2)P_2^6(d-4)P_2(d-6),\\
P_{10}(d)&\sim P_2^5(d)P_2^{20}(d-2)P_2^{21}(d-4)P_2^8(d-6)P_2(d-8),\notag
\end{align}
where $\sim$ stands for equality of determinants (but not of operators!) and $d$ is such
that the final factor is never $P_2(1)$.

Equation \eqref{eq7} is a second computational formula, although it is not very efficient. It
can, however, be used to express the determinants in terms Riemann $\zeta$-function
since such expressions are known for the $P_2$ s. For example, for the
Paneitz--Fradkin--Tseytlin--Riegert\footnote{This designation is often reserved just for the
four dimensional case.} operator on the $5$-- and $7$--spheres,
\begin{align}\label{eq8}
\log\det P_4(5)&=\frac{1}{32}\left(7\log 2-13\frac{\zeta(3)}{\pi^2}+\frac{15}{2}
\frac{\zeta(5)}{\pi^4}\right)\approx0.104642,\notag\\
\log\det P_4(7)&=\frac{-1}{256}\left(3\log 2+\frac{79}{30}\frac{\zeta(3)}{\pi^2}-
\frac{55}{2}\frac{\zeta(5)}{\pi^4}+\frac{63}{4}\frac{\zeta(7)}{\pi^6}\right)
\approx-0.008297,
\end{align}
which, of course, check numerically against the quadrature, \eqref{eq4}. These expressions
can also be deduced from other representations of the determinants. In the next section a
general formula, Theorem 9, is derived from a complex analytic treatment of the integral,
as mentioned in \cite{5}. This makes the above discussion more precise.

\section{An explicit formula}\label{secgf}
In this section, we find a general formula for $\log\det P_{2k}(d)$ ($2k\leq d$) by
attempting to close the contour in the upper half plane and using residues, {\it cf}, \cite{9}
for some similar specific cases. In detail, set $R=4\pi R'$ with $R'\in\mathbb{N}$ and
define the contour $\gamma=\gamma_1+\gamma_2$ to be the horizontal line
$\gamma_1$ from $-R$ to $R$ and the half circle $\gamma_2$ connects $R$, $Ri$ and
$-R$ in the complex plane, as described in Figure \ref{fig1}:
\begin{center}
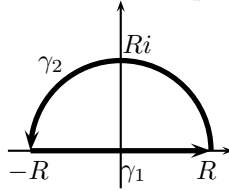
\begin{figure}[htp]
\begin{pspicture}(0,.7)(3,2.5)
\psline{->}(0,1)(3,1)
\psline{->}(1.5,.5)(1.5,3)
\psarc[linewidth=2pt]{->}(1.5,1){1.2}{0}{180}
\psline[linewidth=2pt]{->}(.3,1)(2.7,1)
\put(2.5,.62){$R$}
\put(0,.62){$-R$}
\put(1.52,2.3){$Ri$}
\put(.4,2.1){$\gamma_2$}\put(1.5,0.65){$\gamma_1$}
\end{pspicture}
\caption{The contour $\gamma$}\label{fig1}
\end{figure}
\end{center}

Although the asymptotic behaviour is obvious we proceed carefully.

\begin{lemma}\label{lem5}
We have
$$\lim_{R'\rightarrow\infty}\oint_{\gamma_2}f(z)dz=0.$$
\end{lemma}
\begin{proof}
Let $f(z)=\frac{1}{(x^2+\pi^2)\cosh^m(z/2)}$. Then, for all $z\in\gamma_2$
with $R'\rightarrow\infty$,
\begin{align*}
|f(z)|&\leq\frac{1}{|z^2|-\pi^2}\cdot\frac{2^m}{|e^{z/2}+e^{-z/2}|^m}
=\frac{2^m}{(R^2-\pi^2)e^{mR/2}|1+e^{-Re^{i\theta}}|^m}.
\end{align*}
Note that
\begin{align*}
|1+e^{-Re^{i\theta}}|&=\sqrt{(1+e^{-2R\cos\theta}+2e^{-R\cos\theta}
\cos(R\sin\theta)}.
\end{align*}
Since $R=4\pi R'$ with $R'\in\mathbb{N}$ and $R'\rightarrow\infty$, we have that
$|1+e^{-Re^{i\theta}}|\geq1$. Thus
$$\left|\oint_{\gamma_2}f(z)dz\right|\leq \int_0^\pi |f(4\pi R'e^{i\theta})|
d\theta\leq \frac{2^{m}R'}{\pi(4R'^2-1)e^{2m\pi R'}}\rightarrow0,$$
which completes the proof.
\end{proof}

We now formalize the route that led to \eqref{eq6}.

\begin{theorem}\label{th1}
We have
$$\log\det P_{2k}(d)=\frac{(-1)^{(d-1)/2+k}\pi}{2^{d-2k}}\sum_{j=0}^{k-1}
\binom{2k-1-j}{j}\left(\frac{-1}{4}\right)^j(f_{d+2j-2k}-f_{d+2+2j-2k}),$$
where $f_m=\pi i\sum_{s\geq0}Res_{z=(2s+1)\pi i} \frac{1}{(z^2+\pi^2)
\cosh^m(z/2)}$.
\end{theorem}
\begin{proof}
Using the fact that $U_n(x)=\sum_{j=0}^{\lfloor n/2\rfloor}(-1)^j
\binom{n-j}{j}(2x)^{n-2j}$, we obtain
\begin{align*}
&\log\det P_{2k}(d)\\
&=\frac{(-1)^{(d-1)/2+k}\pi}{2^{d-1}}\int_0^\infty\left(\frac{\sinh^2(x/2)}
{(x^2+\pi^2)\cosh^{d+1}(x/2)}\sum_{j=0}^{k-1}(-1)^j\binom{2k-1-j}{j}
(2\cosh(x/2))^{2k-1-2j}\right)dx\\
&=\frac{(-1)^{(d-1)/2+k}\pi}{2^{d-2k}}\sum_{j=0}^{k-1}\binom{2k-1-j}{j}
\left(\frac{-1}{4}\right)^j\int_0^\infty\frac{\cosh^2(x/2)-1}{(x^2+\pi^2)
\cosh^{f+2+2j-2k}(x/2)}dx\\
&=\frac{(-1)^{(d-1)/2+k}\pi}{2^{d-2k}}\sum_{j=0}^{k-1}\binom{2k-1-j}{j}
\left(\frac{-1}{4}\right)^j(f_{d+2j-2k}-f_{d+2+2j-2k}),
\end{align*}
where $f_m=\int_0^\infty\frac{dx}{(x^2+\pi^2)\cosh^m(x/2)}=\frac{1}{2}
\int_{-\infty}^\infty\frac{dx}{(x^2+\pi^2)\cosh^m(x/2)}$.
By Lemma \ref{lem5}, we have
$$f_m=\frac{1}{2}\lim_{R'\rightarrow\infty}\oint_\gamma\frac{dz}{(z^2+\pi^2)
\cosh^m(z/2)}dz.$$
By the residue theorem, we obtain
$$f_m=\pi i\sum_{s\geq0}Res_{z=(2s+1)\pi i} \frac{1}{(z^2+\pi^2)\cosh^m(z/2)},$$
as required.
\end{proof}

%

\begin{lemma}\label{lem2}
The Laurent series of $\sech^m(z/2)$  at $z=(2s+1)\pi i$ is given by
$$\sech^m(z/2)=(-1)^{ms}\sum_{k\geq0} (-1)^k\left(\frac{i}{2}\right)^{2k-m}
\,D^{(m)}_{2k}\,(z-(2s+1)\pi i)^{2k-m},$$
where the $D^{(m)}_{2k}$ are related to the higher Bernoulli functions by, N\"orlund
\cite{20} p.130,
$$D^{(m)}_{2k}\equiv 2^{2k}\,B^{(m)}_{2k}(m/2)$$
\end{lemma}
\begin{proof}. The proof follows immediately from the standard series
$$(t\cosec t)^m=\sum_{k\geq0}(-1)^k\frac{D^{(m)}_{2k}}{(2k)!}\,t^{2k},\quad|t|<\pi.$$
 (See \cite{19} p.196.)
\end{proof}

The $D^{(m)}_{2k}$ can be evaluated in a number of ways. A direct one is by composition
with the $m=1$ series whose coefficients are ordinary Bernoulli numbers, $B_n$. This gives
$$D^{(m)}_{2n}=\frac{1}{n} \sum_{j=1}^n\binom{2n}{2j}\big((m+1)j-n\big)(2-4^j)\,B_{2j}
\,D^{(m)}_{2n-2j}$$
for all $n\geq1$. Some values are tabulated by N\"{orlund} \cite{20}
and for convenience a short table is included here.

\begin{table}[htp]
\begin{tabular}{l||ccccccc}
  $m\backslash k$ & $k=0$ & $k=1$ & $k=2$ & $k=3$ & $k=4$ &$k=5$&$k=6$\\\hline\hline
  &&&&&&&\\[-6pt]
  $m=1$ & $1$&$-\frac{1}{3}$&$\frac{7}{15}$&$-\frac{31}{21}$&$\frac {127}{15}$&$-\frac{2555}{33}$&$\frac{1414477}{1365}$\\[6pt]
  $m=2$ & $1$&$-\frac{2}{3}$&$\frac{8}{5}$&$-\frac{160}{21}$&$\frac {896}{15}$&$-\frac{7680}{11}$&$\frac{15566848}{1365}$\\[6pt]
  $m=3$ & $1$&$-1$&$\frac {17}{5}$&$-\frac {457}{21}$&$\frac {3287}{15}$&$-\frac {34851}{11}$&$\frac {16954277}{273}$\\[6pt]
  $m=4$ & $1$&$-\frac{4}{3}$&$\frac {88}{5}$&$-\frac {992}{21}$&$\frac {5248}{9}$&$-\frac {111104}{11}$&$\frac {21157888}{91}$\\[6pt]
  $m=5$ & $1$&$-\frac{5}{3}$&$9$&$-\frac {1835}{21}$&$\frac {11513}{9}$&$-\frac {284685}{11}$&$\frac {62451523}{91}$\\[6pt]
\end{tabular}
\caption{Values of $D_{2k}^{(m)}$ for $m=1,2,\ldots,5$ and $k=0,1,\ldots,6$.} \label{tabd}
\end{table}


Next, note that the Laurent series of $\frac{1}{z^2+\pi^2}$ at $z=-\pi i$ is given by
$\sum_{n\geq0}\frac{(-i)^{n-1}}{2^{n+1}\pi^{n+1}}(z+\pi i)^{n-1}$ and at $z=\pi i$
is given by $\sum_{n\geq0}\frac{i^{n-1}}{2^{n+1}\pi^{n+1}}(z-\pi i)^{n-1}$. Hence,
by Lemma \ref{lem2}, we obtain

\begin{lemma}\label{lem3}
Let $i^2=-1$. Then
$$Res_{z=\pi i}\,\frac{\sech^{m}(z/2)}{z^2+\pi^2}=\sum_{n=0}^{\lfloor m/2\rfloor}
\frac{(-1)^n}{(2n)!}\frac{D^{(m)}_{2n}}{2i\pi^{m-2n+1}},$$
$$Res_{z=-\pi i}\frac{\sech^m(z/2)}{z^2+\pi^2}=-\sum_{n=0}^{\lfloor m/2\rfloor}
\frac{(-1)^n}{(2n)!}\frac{D^{(m)}_{2n}}{2i\pi^{m-2n+1}}.$$
\end{lemma}

Note that the Taylor series of $\frac{1}{z^2+\pi^2}$ at $z=(2s+1)\pi i$ with $s\neq 0,-1$
is given by
$$\sum_{n\geq0}\frac{i^{n+2}((s+1)^{n+1}-s^{n+1})}{2^{n+2}\pi^{n+2}
s^{n+1}(s+1)^{n+1}}\,\big(z-(2s+1)\pi i\big)^n.$$
Hence, by Lemma \ref{lem2}, we obtain

\begin{lemma}\label{lem4}
Let $s\in\mathbb{Z}$ with $s\neq 0,-1$. Let $i^2=-1$. Then
\begin{align*}
&Res_{z=(2s+1)\pi i}\frac{\sech^m(z/2)}{z^2+\pi^2}=(-1)^{ms}i
\sum_{n=0}^{\lfloor (m-1)/2\rfloor} \frac{(-1)^n}{(2n)!}
 \frac{D^{(m)}_{2n}}{2\pi^{m-2n+1}}
\left(\frac{1}{s^{m-2n}}
-\frac{1}{(s+1)^{m-2n}}\right).
\end{align*}
\end{lemma}

Hence, by Theorem \ref{th1}, Lemma \ref{lem3} and Lemma \ref{lem4}, we
obtain the following result.

\begin{theorem}\label{th2}
Let $d\geq2k$. We have
$$\log\det P_{2k}(d)=\frac{(-1)^{(d-1)/2+k}\pi}{2^{d-2k}}\sum_{j=0}^{k-1}
\binom{2k-1-j}{j}\left(\frac{-1}{4}\right)^j(f_{d+2j-2k}-f_{d+2+2j-2k}),$$
where
$f_m=\pi i\sum_{s\geq0}Res_{z=(2s+1)\pi i} \frac{\sech^m(z/2)}{(z^2+\pi^2)}$,
$$Res_{z=\pi i}\frac{\sech^m(z/2)}{z^2+\pi^2}=\sum_{n=0}^{\lfloor m/2\rfloor}
\frac{(-1)^n}{(2n)!}\frac{D^{(m)}_{2n}}{2i\pi^{m-2n+1}}$$
and
\begin{align*}
&Res_{z=(2s+1)\pi i}\frac{\sech^m(z/2)}{z^2+\pi^2}=(-1)^{ms}i
\sum_{n=0}^{\lfloor (m-1)/2\rfloor}  \frac{(-1)^n}{(2n)!}\frac{D^{(m)}_{2n}}
{2\pi^{m-2n+1}}\left(\frac{1}{s^{m-2n}}-\frac{1}{(s+1)^{m-2n}}\right).
\end{align*}
\end{theorem}

\begin{remark}
Let $g_{m,s}=Res_{z=\pi i}\frac{\sech^m(z/2)}{z^2+\pi^2}$. Then, by Lemmas
\ref{lem3} and \ref{lem4}, we obtain that
$$|g_{m,0}|>|g_{m,1}|>|g_{m,2}|\cdots\mbox{ and }|g_{m+2,s}|<|g_{m,s}|.$$
\end{remark}

In order to simplify the expression in the statement of Theorem \ref{th2}, we need the
following definition and the following lemma. We bring in the Dirichlet $\eta$--function,
$\eta(s)=\sum_{n\geq1}\frac{(-1)^n}{n^{s}}$. Then, one has, ($\ell\in \mathbb{Z}$),
$$\eta(\ell)=
\left\{\begin{array}{ll}
-\log(2)&\ell=1,\\
(2^{1-\ell}-1)\zeta(\ell)&\ell>1,
\end{array}\right.$$

\begin{lemma}\label{lem6}
Let $f_m=\pi i\sum_{s\geq0}Res_{z=(2s+1)\pi i} \frac{\sech^m(z/2)}{(z^2+\pi^2)}$.
For all $m\geq0$,
\begin{align*}
f_{2m}&=\frac{1}{2}\frac{(-1)^m}{(2m)!}D^{(2m)}_{2m},\\
f_{2m+1}&=-\sum_{n=0}^{m} \frac{(-1)^n}{(2n)!}
D^{(2m+1)}_{2n}\frac{\eta(2m-2n+1)}{\pi^{2m-2n+1}}.
\end{align*}
\end{lemma}
\begin{proof}
By Lemma \ref{lem3} and \ref{lem4}, we have that
\begin{align*}
f_{2m}&=\pi i\sum_{s\geq0}Res_{z=(2s+1)\pi i}\frac{\sech^{2m}(z/2)}{z^2+\pi^2}\\
&=\sum_{n=0}^m \frac{(-1)^n}{(2n)!}D^{(2m)}_{2n} \frac{1}{2\pi^{2m-2n}}-
\sum_{n=0}^{m-1} \frac{(-1)^n}{(2n)!}D^{(2m)}_{2n} \frac{1}{2\pi^{2m-2n}}\\&=
\frac{(-1)^m}{(2m)!}\frac{1}{2}D^{(2m)}_{2m}.
\end{align*}
and
\begin{align*}
f_{2m+1}&=\pi i\sum_{s\geq0}Res_{z=(2s+1)\pi i}\frac{\sech^{2m+1}(z/2)}
{z^2+\pi^2}\\
&=\sum_{n=0}^m \frac{(-1)^n}{(2n)!}D^{(2m+1)}_{2n} \frac{1}{2\pi^{2m-2n+1}}\\
&-\sum_{s\geq1}\left[(-1)^{s}\sum_{n=0}^m \frac{(-1)^n}{(2n)!}D^{(2m+1)}_{2n}
\frac{1}{2\pi^{2m-2n+1}}
\left(\frac{1}{s^{2m+1-2n}}-\frac{1}{(s+1)^{2m+1-2n}}\right)\right],
\end{align*}
which, by the definition of the function $\eta(\ell)$, implies
\begin{align*}
f_{2m+1}&=\sum_{n=0}^m \frac{(-1)^n}{(2n)!}D^{(2m+1)}_{2n}
\frac{-1}{2\pi^{2m-2n+1}}+\sum_{n=0}^m \frac{(-1)^n}{(2n)!}D^{(2m+1)}_{2n}
\frac{1-2\eta(2m-2n+1)}{2\pi^{2m-2n+1}}\\
&=-\sum_{n=0}^m \frac{(-1)^n}{(2n)!}D^{(2m+1)}_{2n}
\frac{\eta(2m-2n+1)}{\pi^{2m-2n+1}},
\end{align*}
as claimed.
\end{proof}

\begin{table}[htp]
\begin{align*}
&f_0=\frac{1}{2},&&f_1=\frac{1}{\pi}\log2\\
&&&\approx0.2206356001,\\
&f_2=\frac{1}{6},&&f_3=\frac{3\zeta(3)}{4\pi^3}+\frac{\log2}{2\pi}\\
&&&\approx0.1393939347,\\
&f_4=\frac{11}{90},&&f_5=\frac{15\zeta(5)}{16\pi^5}+\frac{5\zeta(3)}{8\pi^3}
+\frac{3\log2}{8\pi}\\
&&&\approx0.1101451199,\\
&f_6=\frac{191}{1890},&&f_7=\frac{63\zeta(7)}{64\pi^7}+\frac{35\zeta(5)}{32\pi^5}
+\frac{259\zeta(3)}{480\pi^3}+\frac{5\log2}{16\pi}\\
&&&\approx0.09390203072,\\
&f_8=\frac{2497}{28350},&&f_9=\frac{255\zeta(9)}{256\pi^9}+\frac{189
\zeta(7)}{128\pi^7}+\frac{141\zeta(5)}{128\pi^5}+\frac{3229\zeta(3)}
{6720\pi^3}+\frac{35\log2}{128\pi}\\   &&&\approx.08321740587.\\[-18pt]
\end{align*}
\caption{Values for $f_m$, where $m=0,1,\ldots,8$.}\label{tab1}
\end{table}

Hence, by Theorem \ref{th2} and Lemma \ref{lem6}, we obtain our main result.
\begin{theorem}\label{th3}
Let $d\geq2k\geq2$ and $d$ an odd number. We have
$$\log\det P_{2k}(d)=\frac{(-1)^{(d-1)/2+k}\pi}{2^{d-2k}}\sum_{j=0}^{k-1}
\binom{2k-1-j}{j}\left(\frac{-1}{4}\right)^j(f_{d+2j-2k}-f_{d+2+2j-2k}),$$
where for all $m\geq0$ (see Table \ref{tab1}),
\begin{align*}
f_{2m}&=\frac{1}{2}\frac{(-1)^m}{(2m)!}\,D^{(2m)}_{2m},\\
f_{2m+1}&=-\sum_{n=0}^{m} \frac{(-1)^n}{(2n)!}
D^{(2m+1)}_{2n}\frac{\eta(2m-2n+1)}{\pi^{2m-2n+1}}.
\end{align*}
\end{theorem}


\begin{example}
Theorem \ref{th3} with $k=1$ and $d>2$, we have for the Yamabe--Penrose operator
$$\log\det P_{2}(d)=\frac{(-1)^{(d+1)/2}\pi}{2^{d-2}}\,(f_{d-2}-f_{d}).$$
For instance,
\begin{align*}
\log\det P_{2}(3)&=\frac{1}{2}\pi(f_1-f_3)=
\frac{1}{4}\log 2 - \frac{3}{8}\frac {\zeta(3)}{\pi^2}\\
&\approx0.1276141094,\\
\log\det P_{2}(5)&=\frac{-\pi}{8}(f_3-f_5)=\frac{15}{128}\frac{\zeta(5)}{\pi^4}
-\frac{1}{64}\frac{\zeta(3)}{\pi^2} -\frac{1}{64}\log2\\
&\approx-0.01148598272,
\end{align*}
which, as a check, agree with standard, and old, expressions.
\end{example}

\begin{example}
Theorem \ref{th3} with $k=2$ and $d>4$, we have for the Paneitz operator
$$\log\det P_{4}(d)=\frac{(-1)^{(d-1)/2}\pi}{2^{d-4}}\,(f_{d-4}-f_{d-2})+
\frac{(-1)^{(d+1)/2}\pi}{2^{d-3}}\,(f_{d-2}-f_d).$$ For instance, agreement is found
with \eqref{eq8} and further
\begin{align*}
 \log\det P_{4}(9)&={ \frac {11}{4096}}\,\log2+{ \frac
{751}{215040}}\,{\frac {\zeta \left( 3 \right) }{{\pi }^{2}}}-{\frac {39}{4096}} \,{\frac
{\zeta \left( 5 \right) }{{\pi }^{4}}} -{\frac {189}{4096}}\,{\frac{\zeta \left( 7 \right)
}{{\pi }^{6}}}+{\frac {255}{8192}}\,{\frac {
\zeta \left( 9 \right) }{{\pi }^{8}}}\\
&\approx0.001070181258,\\
\log\det P_{4}(11)&=-{\frac {13}{65536}}\,\log2-{\frac {2867}{9830400}}\,
{\frac {\zeta \left( 3 \right) }{{\pi }^{2}}}+
{\frac {737}{4128768}}\,{\frac {\zeta \left( 5 \right) }{{\pi }^{4}}}+
{\frac {1911}{655360}}\,{\frac {\zeta \left( 7 \right) }{{\pi }^{6}}}
+{\frac {595}{131072}}\,{\frac {\zeta \left( 9 \right) }{{\pi }^{8}}}\\
&-{\frac {1023}{262144}}\,{\frac {\zeta \left( 11 \right) }{{\pi }^{10}}}\\
&\approx-0.0001676200873,\\
\log\det P_{4}(13)&={\frac {35}{1048576}}\,\log2+{\frac {189349}{3633315840}}\,
{\frac {\zeta \left( 3 \right) }
{{\pi }^{2}}}+{\frac {39701}{1981808640}}\,{\frac {\zeta \left( 5 \right) }{{\pi}^{4}}}
-{\frac {1115}{3145728}}\,{\frac {\zeta \left( 7 \right) }{{\pi }^{6}}}\\
&-{\frac {6613}{6291456}}\,{\frac {\zeta \left( 9 \right) }{{\pi }^{8}}}-
{\frac {1705}{2097152}}\,{\frac {\zeta \left( 11 \right) }{{\pi }^{10}}}+
{\frac {4095}{4194304}}\,{\frac {\zeta \left( 13 \right) }{{\pi }^{12}}}\\
&\approx0.00002920638544.
\end{align*}
\end{example}

\begin{example}
Theorem \ref{th3} with $k=3$ and $d>6$, we have
$$\log\det P_{6}(d)=
\frac{(-1)^{(d+1)/2}\pi}{2^{d-6}}(f_{d-6}-f_{d-4})+\frac{(-1)^{(d-1)/2}\pi}
{2^{d-6}}(f_{d-4}-f_{d-2})+\frac{3(-1)^{(d+1)/2}\pi}{2^{d-2}}(f_{d-2}-f_d).$$
For instance,
\begin{align*}
\log\det P_{6}(7)&={\frac {99}{512}}\,\log2 -{\frac {2199}{5120}}\,{
\frac {\zeta \left( 3 \right) }{{\pi }^{2}}}+{\frac {465}{1024}}\,{
\frac {\zeta \left( 5 \right) }{{\pi }^{4}}}-{\frac {189}{2048}}\,{
\frac {\zeta \left( 7 \right) }{{\pi }^{6}}}\\
&\approx0.08645416332,\\
\log\det P_{6}(9)&=-{\frac {143}{16384}}\,\log2-{\frac {5447}{860160}}\,
{\frac {\zeta \left( 3 \right) }{{\pi }^{2}}}+{\frac {1603}{16384}}\,{
\frac {\zeta \left( 5 \right) }{{\pi }^{4}}}-{\frac {1827}{16384}}\,{
\frac {\zeta \left( 7 \right) }{{\pi }^{6}}}+{\frac {765}{32768}}\,{
\frac {\zeta \left( 9 \right) }{{\pi }^{8}}}\\
&\approx-0.005894056955,\\
\log\det P_{6}(11)&={\frac {117}{131072}}\,\log2+{\frac {49451}{45875200}}
\,{\frac {\zeta  \left( 3 \right) }{{\pi
}^{2}}}-{\frac {12283}{2752512}}\,{\frac {\zeta  \left( 5 \right) }{{\pi }^{4
}}} -{\frac {21987}{
1310720}}\,{\frac {\zeta  \left( 7 \right) }{{\pi }^{6}}}\\
&+{\frac {6885
}{262144}}\,{\frac {\zeta  \left( 9 \right) }{{\pi }^{8}}}-{\frac {
3069}{524288}}\,{\frac {\zeta  \left( 11 \right) }{{\pi }^{10}}}\\
&\approx0.0006876310510,\\
\log\det P_{6}(13)&=-{\frac {255}{2097152}}\,\log2-{\frac {414199}{2422210560}}
\,{\frac {\zeta  \left( 3
\right) }{{\pi }^{2}}}+{\frac {314341}{1321205760}}\,{\frac {\zeta  \left( 5 \right) }{{
\pi }^{4}}}+{\frac {4513}{2097152}}\,{\frac {\zeta  \left( 7 \right) }{{\pi }^{6}}}\\
&+{\frac {9027}{4194304}}\,{\frac {\zeta  \left( 9 \right) }{{\pi }^{8}
}}-{\frac {25575}{4194304}}\,{\frac {\zeta  \left( 11 \right) }{{\pi }
^{10}}}+{\frac {12285}{8388608}}\,{\frac {\zeta  \left( 13 \right) }{{
\pi }^{12}}}\\
&\approx-0.0001001554942.
\end{align*}

\end{example}
\section{Central quantities}

An expression for $f_{2m+1}$ of the same structure as that in Lemma 9 was derived in the
Appendix to \cite{5} by a different manipulation of the basic integral appearing in Theorem
2. Using another expansion of $\sech^m$, it yielded the sum of zeta functions and $\log2$
with coefficients given, up to constants, by `central differentials of nothing', \cite{22}, or,
equivalently, by central factorial coefficients of the first kind, \cite{13,22}. It is repeated
here
\begin{align*}
f_{2m+1}=
(-1)^m\,D0^{[2m+1]}\,\log2+
(-1)^m\sum_{n=1}^m{\frac{(-1)^{-n}\,2^{2(m-n)}\,D^{2n+1}0^{[2m+1]}}
{\pi^{2n+1}(2m)!(2n+1)}}\,(1-2^{-2n})\,\zeta(2n+1).
\end{align*}

A comparison with Lemma 9 shows that the residue calculation has yielded an expression for
the central differentials of nothing in terms of generalised Bernoulli numbers.

Expressed in terms of the central factorial coefficients of the first kind, $t(*,*)$,
\cite{14,15}, this formula is
\begin{align}
t(2m+1,2n+1)=2^{2(n-m)}\binom{2m}{2n}\,\,D^{(m)}_{2m-2n},
\end{align}
which is a known relation, see Liu, \cite{17}, who obtains it by comparing recursions. The
calculation of the central quantities is thus entirely equivalent to that of the higher
$D$--N\"{o}rlund numbers and equation (9) can be read numerically in either
direction.\footnote{ An early table of the $t(*,*)$ can be found in Thiele, \cite{21}. \cite{15}
has more extensive lists, but they are easy to calculate {\it ab initio} by machine.}

An additional point is that the present method bypasses Jensen's integral for the
$\eta$--function, invoked in \cite{5}. Hence the present result can lead to a {\it derivation}
of this formula, for odd arguments. Jensen's method also involves Cauchy's theorem,
\cite{16} p.103. A useful survey of these representations is given by Milgram, \cite{18}.

\section{Graphs}

Finally, in order to give a visualisation of the numbers, Fig.2 and Fig.3  plot $P_{2k}(35)$
for the allowed $k$, {\it ie} \ $1\le k\le17$. (There is nothing special about the dimension
35.) The values oscillate about zero with increasing amplitude, maximum to minimum being
about $10^{10}$.

 \includegraphics[scale=.75]{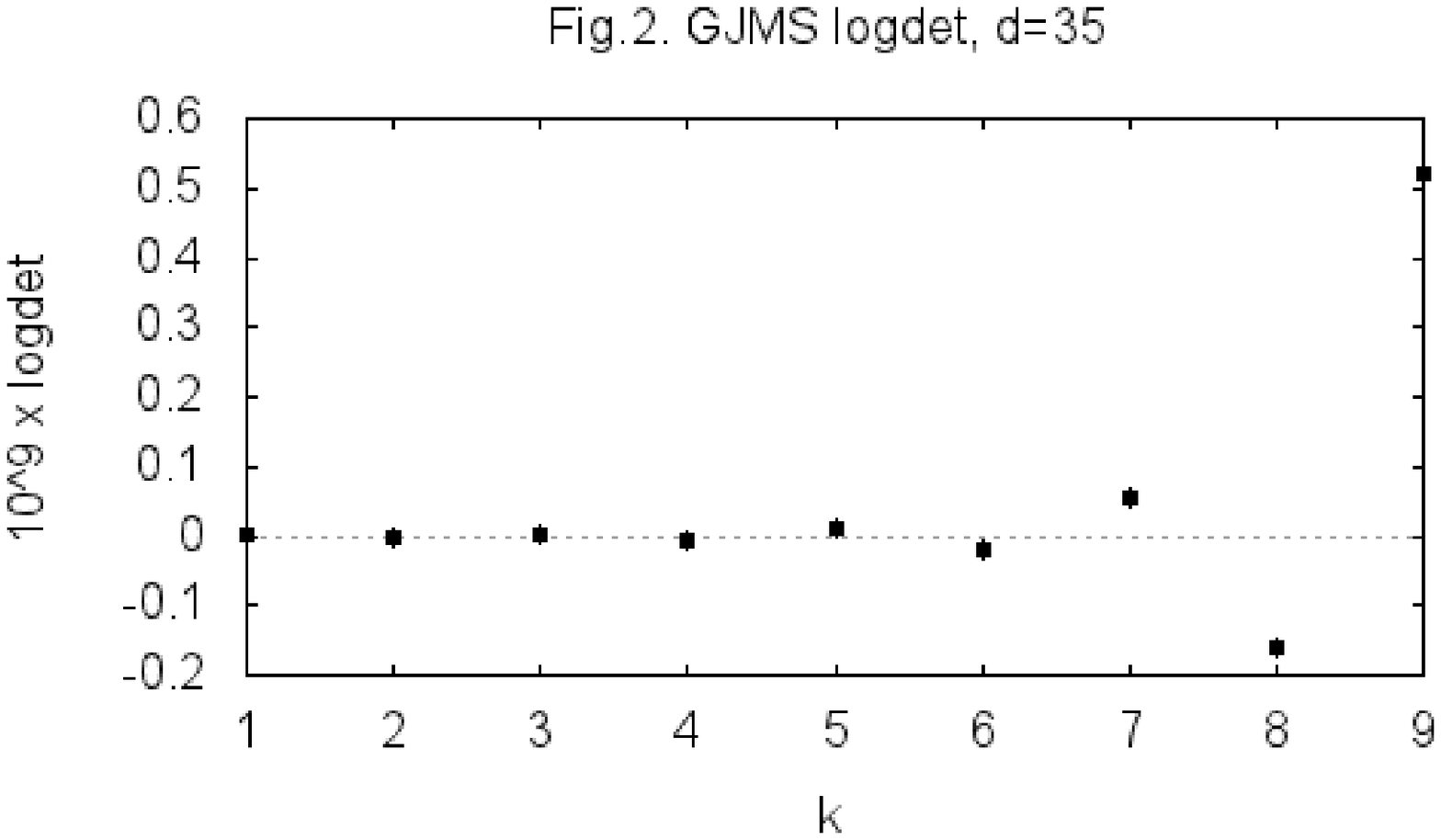}

 \includegraphics[scale=.75]{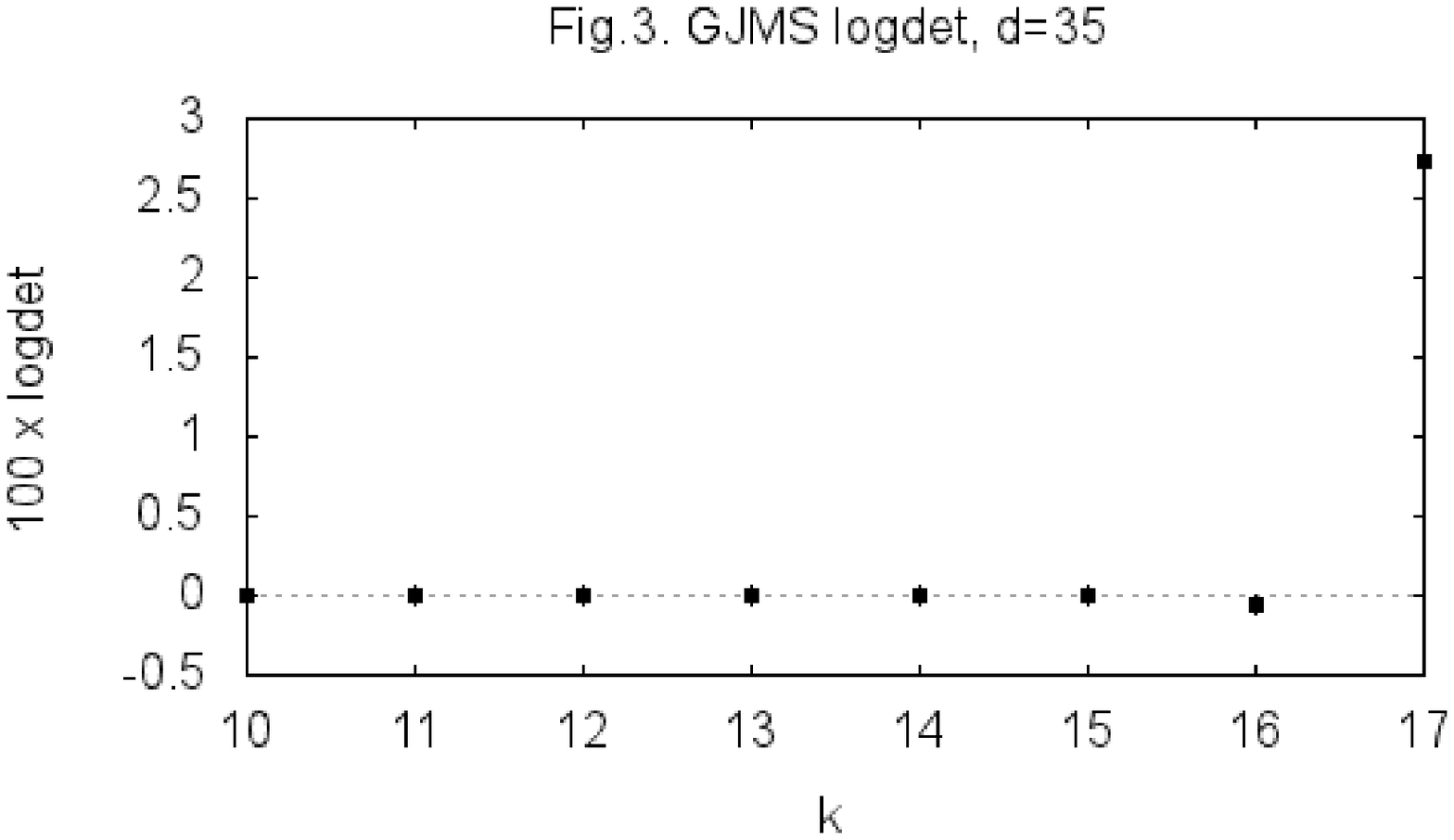}

 \includegraphics[scale=.75]{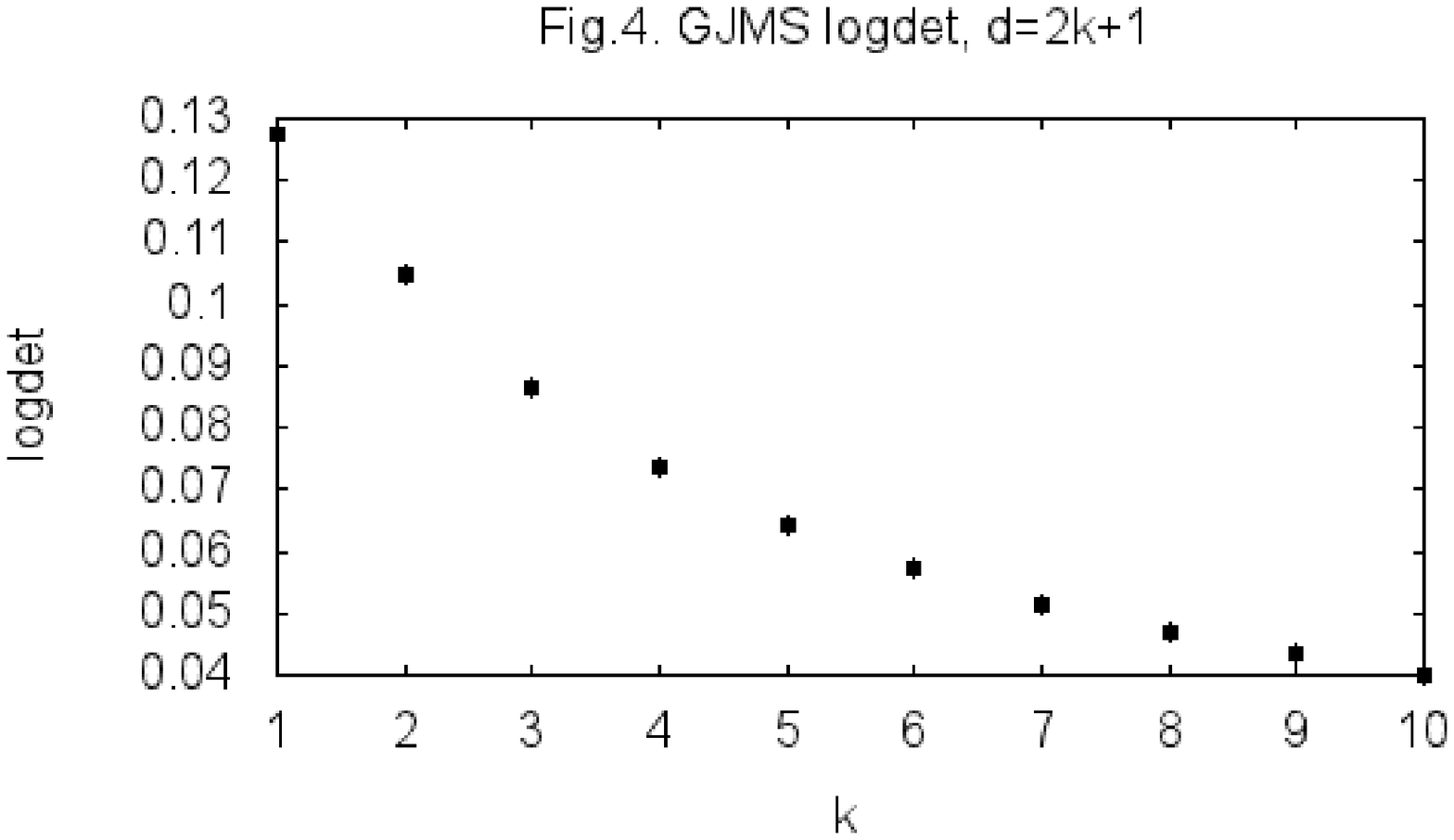}

 \includegraphics[scale=.75]{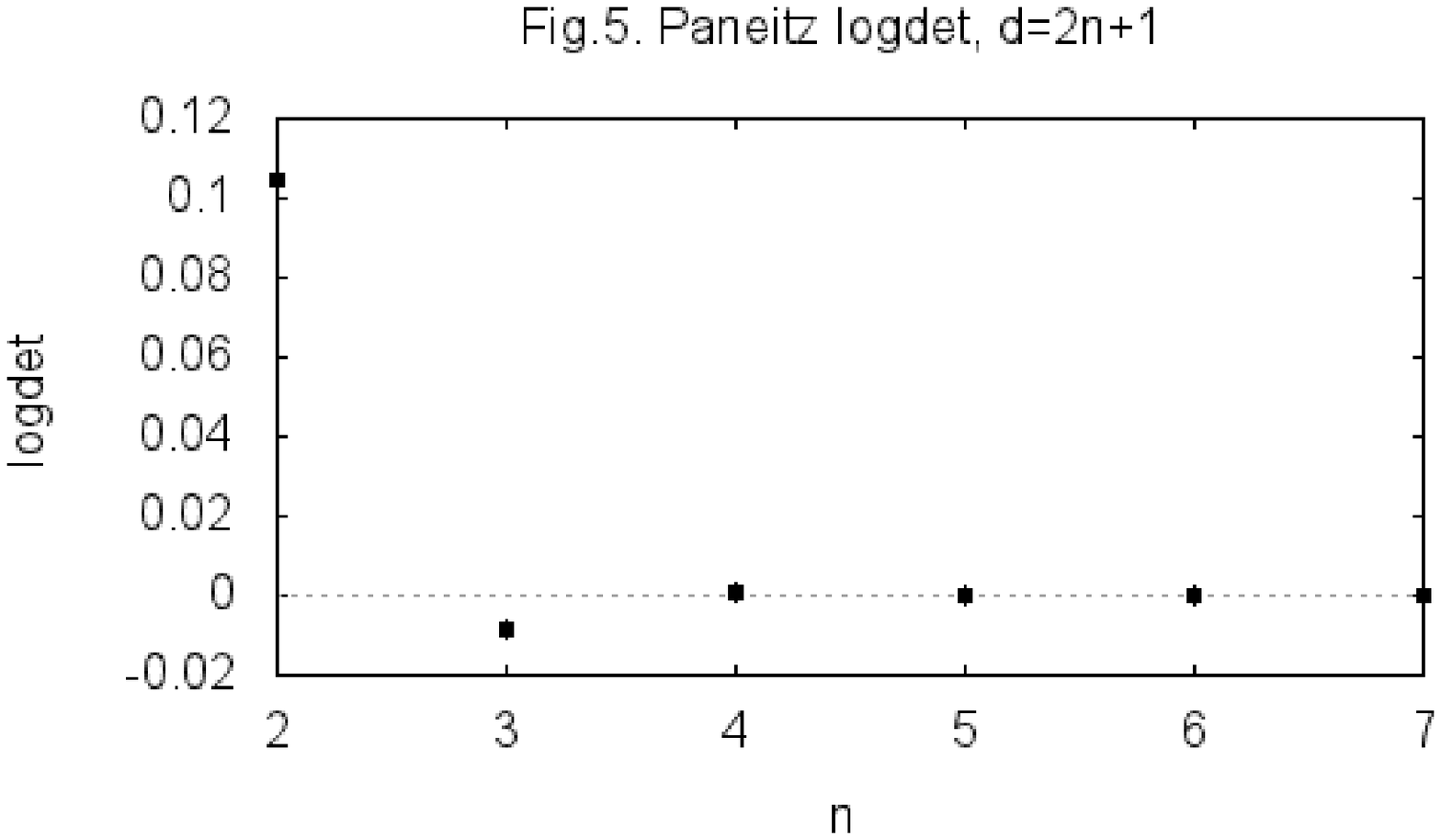}

Fig.4 shows the logdets at the limiting $k=(d-1)/2$ for $d=3$ to $d=21$ and Fig.5 the
logdet of the Paneitz operator ($k=2$) against dimension. The values oscillate about zero
with decreasing amplitude so that the determinant tends to unity with increasing $d$, as is
clear from \eqref{eq4}. This behaviour is typical.

\section{Conclusion and remarks}
The main results of this work are the quadrature \eqref{eq4} (equivalently \eqref{eq5}),
the product formula, \eqref{eq6} and the explicit form, Theorem 9, for the determinant of
the scalar GJMS operator on an odd sphere.

A curious coincidence is that the GJMS operator on a sphere can be written as a central
factorial, (1), and that central quantities appear in its determinant.

The method could be extended to the Dirac case, and to higher spins.

The computation for even spheres is somewhat harder and involves a multiplicative
anomaly, calculated in \cite{1}.

\bibliographystyle{amsplain}

\end{document}